\documentclass[11pt]{article}

\usepackage{amsfonts,amssymb,amsmath,amscd,amsthm,latexsym,color}
\usepackage{natbib}
\bibpunct{(}{)}{,}{a}{,}{,}
\usepackage{graphicx}
\usepackage{fullpage}
\usepackage{cancel}
\usepackage{algorithm,algorithmic}
\usepackage{wasysym}

\newtheorem{theorem}{Lemma}
\newtheorem{example}{Example}

\newcommand{\bu}{\mathbf{u}}

\newcommand{\by}{\mathbf{y}}
\newcommand{\bz}{\mathbf{z}}

\newcommand{\vs}{\vspace{0.5cm}}

\newcommand{\btheta}{\boldsymbol{\theta}}
\newcommand{\feta}{\boldsymbol{\eta}}

\usepackage{geometry}
\usepackage{setspace} 

\title{Why approximate Bayesian computational (ABC) methods cannot handle model choice problems}

\author{
{\sc Christian P.~Robert} \\
Universit\'e Paris Dauphine, CEREMADE,\\ 
IUF, and CREST\\
\and
{\sc Jean-Michel Marin} \\
I3M, UMR CNRS 5149 \\
Universit\'e Montpellier 2\\
\and
{\sc Natesh S.~Pillai} \\
Department of Statistics, Harvard University 
}

\date{}

\begin{document}

\maketitle

\begin{abstract}

Approximate Bayesian computation (ABC), also known as likelihood-free methods, have become a favourite tool for
the analysis of complex stochastic models, primarily in population genetics but also in financial analyses. We
advocated in \cite{grelaud:marin:robert:rodolphe:tally:2009} the use of ABC for Bayesian model choice in the
specific case of Gibbs random fields (GRF), relying on a sufficiency property mainly enjoyed by GRFs to show
that the approach was legitimate. Despite having previously suggested the use of ABC for model choice in a
wider range of models in the DIY ABC software \citep{cornuet:santos:beaumont:etal:2008}, we present theoretical
evidence that the general use of ABC for model choice is fraught with danger in the sense that no amount of
computation, however large, can guarantee a proper approximation of the posterior probabilities of the models
under comparison. 


\vs \noindent \textbf{Keywords}: likelihood-free methods, Bayes factor, DIYABC, Bayesian model choice, sufficiency.

\end{abstract}

\section{Introduction}
Inference on population genetic models such as coalescent trees is one representative example of cases when
statistical analyses like Bayesian inference cannot operate because the likelihood function associated with the
data is not completely known, i.e.~cannot be computed in a manageable time
\citep{tavare:balding:griffith:donnelly:1997,beaumont:zhang:balding:2002,cornuet:santos:beaumont:etal:2008}.
The fundamental reason for this impossibility is that the statistical model associated with coalescent data
needs to integrate over trees of extreme complexity.

In such settings, traditional approximation tools based on Monte Carlo simulation \citep{robert:casella:2004}
from the Bayesian posterior distribution are unavailable for all practical purposes. Indeed, due to the
complexity of the latent structures defining the likelihood (such as the coalescent tree), simulation of those
structures is too unstable to be trusted to bring a reliable approximation in a manageable time. Such complex
models call for a practical if cruder approximation method, the ABC methodology being a serious contender,
where ABC stands for {\em approximate Bayesian computation}.  \cite{tavare:balding:griffith:donnelly:1997} and
\cite{pritchard:seielstad:perez:feldman:1999} introduced ABC methods as a rejection technique bypassing the
computation of the likelihood function via a simulation from the corresponding distribution.  For recent
reviews on ABC, see \cite{beaumont:2010} and \cite{lopes:beaumont:2010}.  The wide and successful array of
applications based on implementations of ABC in genomics and ecology is covered by
\cite{csillery:blum:gaggiotti:francois:2010}, while the number of publications relying on this technique runs
in the hundreds. 

\cite{pritchard:seielstad:perez:feldman:1999} describe the use of model choice based on ABC for distinguishing
between different mutation models.  The intuition behind the method is that the average ABC acceptance rate
associated with a given model is proportional to the marginal likelihood corresponding to this approximative
model, when identical summary statistics, distance, and tolerance level are used for all models. In practice,
an estimate of the ratio of marginal likelihoods is given by the ratio of observed acceptance rates. Using
Bayes formula, estimates of the posterior probabilities are straightforward to derive.  This approach has been
widely used in the literature (see, e.g., \citealp{estoup:etal:2004}, \citealp{miller:etal:2005}, and
\citealp{pascual:etal:2007}, \citealp{sainudiin:etal:2011}). Note that \cite{miller:etal:2005} is particularly
influencial for the conclusion it derives from the ABC analysis: the focus of this {\em Science} paper is the
European invasion of the western corn rootworm, which is North America's most destructive corn pest. Because this
pest was initially introduced in Central Europe, it was believed that subsequent outbreaks in Western Europe
originated from this area.  Based on this ABC model choice analysis of the genetic variability of the rootworm,
the authors conclude that this belief is false: There have been at least three independent introductions from
North America during the past two decades.

An improvement to the above estimate is due to \cite{fagundes:etal:2007}, thanks to a regression
regularisation. In this approach. model indices are processed as categorical variables in a formal multinomial
(polychotomous) regression. For instance, when comparing two models, this leads to a standard logistic
regression.
Rejection-based approaches were lately introduced by \cite{cornuet:santos:beaumont:etal:2008},
\cite{grelaud:marin:robert:rodolphe:tally:2009} and \cite{toni:etal:2009}, in a Monte Carlo perspective
simulating model indices as well as model parameters. Those more recent extensions   are already widely in use
by the population genetics community, as exemplified by
\cite{belle:etal:2008,cornuet:ravigne:estoup:2010,excoffier:leuenberger:wegman:2009,ghirotto:etal:2010,guillemaud:etal:2009,leuenberger:wegmann:2010,
patin:etal:2009,ramakrishnan:hadly:2009,verdu:etal:2009}, or \cite{wegmann:excoffier:2010}. Another
illustration of the popularity of this approach is given by the availability of three 
three softwares implementing an ABC model choice methodology: 
\begin{itemize}
\item ABC-SysBio\footnote{{\sf http://abc-sysbio.sourceforge.net}}, developped by 
the Theoretical Systems Biology Group at Imperial College London, which implements
a SMC-based ABC for inference in system biology, including model-choice \citep{toni:etal:2009}.
\item DIYABC\footnote{{\sf http://www1.montpellier.inra.fr/CBGP/diyabc}}, developped by
the Centre de Biologie et de Gestion des Populations, at INRA Montpellier, which implements
a regularised ABC-MC algorithm on population history using molecular markers
\citep{cornuet:santos:beaumont:etal:2008}.
\item PopABC\footnote{{\sf http://code.google.com/p/popab}}, developped by
the School of Biological Sciences at the University of Reading, which implements
a regular ABC-MC algorithm for genealogical simulation \citep{lopes:balding:beaumont:2009}.
\end{itemize}

\cite{grelaud:marin:robert:rodolphe:tally:2009} process via ABC the specific case of Gibbs random fields with
missing normalising constants. They establish that exact Bayesian model selection can be implemented in this
setting, deriving this result from the property that the concatenation of the sufficient statistics across
models is also sufficient for model comparison.  In a subsequent paper,
\cite{didelot:everitt:johansen:lawson:2011} advocate the role of ABC approximations in general Bayesian model
choice. The issue of sufficiency is covered in this paper, with a generic cross-model sufficiency completion
leading the authors to validate the method in full generality, including in-sufficient cases.

In this paper, we argue that ABC is a valid approximation method for conducting Bayesian inference in complex
stochastic models, barring the limitation that it cannot discriminate between those complex stochastic models
when based on summary statistics.  In essence, we highlight the fact that, since ABC is conducting model choice
based on in-sufficient statistics, the resulting inference is flawed in that the loss of information is severe
to the point of inconsistency, namely that the ABC model selection cannot recover the proper model, even with
an infinite amount of observation and computation.  We demonstrate
this inconsistency in the limiting (and more favourable) case of sufficient statistics.

The conclusion of the current paper are thus quite negative in that we consider that conducting testing or
model comparison using ABC does not carry any reliable weight of evidence and therefore should not be trusted.
More empirical measures such as those proposed in \cite{ratmann:andrieu:wiujf:richardson:2009} 
and \cite{drovandi:pettitt:faddy:2011} seem to be the only possibility at the current time for 
conducting model comparison. We are therefore at odds with the positive conclusion found in 
\cite{didelot:everitt:johansen:lawson:2011}, as discussed below. 

We stress here that, while \cite{templeton:2008,templeton:2010} repeatedly expressed reservations about the
formal validity of the ABC approach in statistical testing, those criticisms were addressed at the Bayesian
paradigm {\em per se} rather than at the approximation method. Quite clearly, Templeton's criticisms got
rebutted in \cite{clade:2010,csillery:blum:gaggiotti:francois:2010b,berger:fienberg:raftery:robert:2010} and
are not relevant for the current paper. 

The plan of the paper is as follows: in Section \ref{Genesis}, we recall the basics of ABC as well as its
justification; Section \ref{TheRamones} exposes why a Bayes factor based on an ABC approximation is not
converging to the true Bayes factor as the computational effort increases; Section \ref{TheClash} explains the
specificity of MRFs in this regard, while Section \ref{Siouxie} illustrates the potential for divergence in
examples. Sectoion \ref{RollingStones} concludes the paper.

\section{The ABC approach and its justifications}\label{Genesis}

The setting in which ABC operates is the approximation of the simulation from the posterior distribution
$\pi(\btheta|\by) \propto \pi(\theta) f(\by|\btheta)$ when both distributions associated with $\pi$ and $f$ can
be simulated.  The first ABC algorithm was introduced by \cite{pritchard:seielstad:perez:feldman:1999} in a
genetic setting, as follows: given a sample $\by$ from a sample space $\mathcal{D}$,

\begin{algorithm}[H]
\caption{ABC sampler\label{algo:ABC0}}
\begin{algorithmic}
\FOR {$i=1$ to $N$}
	\REPEAT
	\STATE Generate $\btheta'$ from the prior distribution $\pi(\cdot)$
        \STATE Generate $\bz$ from the likelihood $f(\cdot|\btheta')$
        \UNTIL {$\rho\{\eta(\bz),\eta(\by)\}\leq \epsilon$}
           \STATE set $\btheta_i=\btheta'$,
\ENDFOR
\end{algorithmic}
\end{algorithm}

\smallskip
The parameters of the ABC algorithm are the statistic $\eta$, the distance
$\rho\{\cdot,\cdot\}\geq 0$, and the tolerance level $\epsilon>0$. The approximation of the
posterior distribution provided by the algorithm is that it samples from the marginal in $\btheta$ of 
the joint distribution
\begin{equation}\label{eq:abctarget}
\pi_\epsilon(\btheta,\bz|\by)=
\frac{\pi(\btheta)f(\bz|\btheta)\mathbb{I}_{A_{\epsilon,\by}}(\bz)}
{\int_{A_{\epsilon,\by}\times\Theta}\pi(\btheta)f(\bz|\btheta)\text{d}\bz\text{d}\btheta}\,,
\end{equation}
where $\mathbb{I}_B(\cdot)$ denotes the indicator function of the set $B$ and where
$$
A_{\epsilon,\by}=\{\bz\in\mathcal{D}|\rho\{\eta(\bz),\eta(\by)\}\leq \epsilon\} \,.
$$
The basic justification of the ABC approximation is that, when using a sufficient statistic $\eta$ and 
a small (enough) tolerance $\epsilon$, we have 
$$
\pi_\epsilon(\btheta|\by)=\int \pi_\epsilon(\btheta,\bz|\by)\text{d}\bz\approx \pi(\btheta|\by)\,,
$$
the (correct) posterior distribution $\pi(\btheta|\by)$
being the limit as $\epsilon$ goes to zero of $\pi_\epsilon(\btheta|\by)$.

In practice, the statistic $\eta$ is not sufficient and the approximation then converges to 
$\pi_\epsilon(\btheta|\eta(\by))$. This fact is appreciated by users in the field who see this
loss of information as an unvoidable price to pay for the access to computable quantities. While acknowledging the
gain brought by ABC in handling Bayesian inference in complex models,
we will demonstrate below that the loss due to the ABC approximation may be arbitrary in the specific setting
of Bayesian model choice and testing, whether or not $\eta$ is sufficient.

\section{ABC and model choice}\label{RoxyMusic} 

Testing and model choice constitute a highly specific domain of Bayesian analysis that involves conceptual and
computational complexification since several models are simultaneously considered
\citep{robert:2001,marin:robert:2010}. Given that both inferential problems are processed the same way in a
Bayesian perspective, we will only mention model choice in the remainder of the paper, but the reader must bear
in mind that we cover testing as a particular case. The standard tool on which a Bayesian approach relies is
the evidence \citep{jeffreys:1939}, also called the marginal likelihood,
$$
w(\by) = \int_\Theta \pi(\btheta) f(\by|\btheta)\,\text{d} \btheta\,,
$$
that leads to the Bayes factor for comparing the evidences brought by the data on models with likelihoods 
$f_1(\bz|\btheta_1)$ and $f_2(\bz|\btheta_2)$,
$$
B_{12}(\by) = \dfrac{w_1(\by)}{w_2(\by)} =\dfrac{\int_{\Theta_1} \pi_1(\btheta_1) f_1(\by|\btheta_1)\,\text{d}
\btheta_1}{\int_{\Theta_2} \pi_2(\btheta_2) f_2(\by|\btheta_2)\,\text{d} \btheta_2}\,.
$$
As detailed in the Bayesian literature \citep{berger:1985,robert:2001,mackay:2002,marin:robert:2010}, this 
ratio provides an absolute criterion for model comparison that is naturally penalised for model complexity
\citep{clade:2010,berger:fienberg:raftery:robert:2010} and whose first order approximation is the Bayesian
information criterion (BIC).

Given that this issue is fundamental to our point, we recall that Bayesian model choice proceeds by creating a
probability structure across models (or likelihoods). Namely, in addition to the parameters associated with
each model, a Bayesian inference introduces the model index $\mathcal{M}$ as an extra parameter. It is
associated with its own prior distribution, $\pi(\mathcal{M}=m)$ ($m=1,\ldots,M$), while the prior distribution
on the parameter is conditional on the value $m$ of the model index, denoted by $\pi_m(\btheta_m)$ and defined
on the parameter space $\Theta_m$.  The choice between those models is then driven by the posterior
distribution of $\mathcal{M}$, 
$$
\mathbb{P}( \mathcal{M} | \by ) = \dfrac{\pi(\mathcal{M}=m) w_m(\by) }{\sum_k \pi(\mathcal{M}=k) w_k(\by) }
$$
where $w_k(\by)$ denotes the marginal likelihood of $\by$ for model $k$.

While this distribution is well-defined and straightforward to interpret, it offers a challenging computational
conundrum in Bayesian analysis. Moreover, the solutions found in the literature
\citep{chen:shao:ibrahim:2000,marin:robert:2010} do not handle the case when the likelihood is not available
and ABC represents the almost unique alternative.

As exposed in e.g. \cite{grelaud:marin:robert:rodolphe:tally:2009}, \cite{toni:stumpf:2010}, and
\cite{didelot:everitt:johansen:lawson:2011},
once $\mathcal{M}$ is incorporated within the parameters, the ABC approximation to the
posterior follows from the same principles as regular ABC.  The corresponding implementation is as follows,
using for the tolerance region a statistic $\feta(\bz)=(\eta_1(\bz),\ldots,\eta_M(\bz))$ that is the
concatenation of the summary statistics used for all models (with an obvious elimination of duplicates).

\begin{algorithm}\caption{ABC model choice sampler (ABC-MC) \label{algo:ABCMoC}} 
\begin{algorithmic} 
\FOR {$i=1$ to $N$} 
\REPEAT 
\STATE Generate $m$ from the prior $\pi(\mathcal{M}=m)$ 
\STATE Generate $\btheta_{m}$ from the prior $\pi_{m}(\btheta_m)$ 
\STATE Generate $\bz$ from the model $f_{m}(\bz|\btheta_{m})$ 
\UNTIL {$\rho\{\feta(\bz),\feta(\by)\}\leq\epsilon$} 
\STATE Set $m^{(i)}=m$ and $\btheta^{(i)}=\btheta_m$ 
\ENDFOR 
\end{algorithmic} 
\end{algorithm} 

The ABC estimate of the posterior probability $\pi(\mathcal{M}=m|\by)$ is then the frequency of acceptances
from model $m$ in the above simulation
$$ \widehat{ \mathbb{P}( \mathcal{M} | \by ) }= \dfrac{1}{N}\,\sum_{i=1}^N \mathbb{I}_{m^{(i)}=m}\,.  $$ 
This also corresponds to the frequency of simulated pseudo-dataset from model $m$ that are closer to the data
$\by$ than the tolerance $\epsilon$. In order to improve the estimation by smoothing,
\cite{cornuet:santos:beaumont:etal:2008} follow the rationale that motivated the use of a local linear
regression in \cite{beaumont:zhang:balding:2002} and rely on a weighted polychotomous logistic regression to
estimate $\pi(\mathcal{M}=m|\by)$.  This modelling is implemented in the DIYABC software.

\section{The difficulty with ABC-MC}\label{TheRamones}

Most perspectives on ABC do not question the role of the ABC distance nor of the statistic $\feta$ in model
choice settings. There is however a much stronger discrepancy between the genuine Bayes factor / posterior
probability and the approximations resulting from ABC.

The ABC approximation to a Bayes factor, $B_{12}$ say, resulting from Algorithm \ref{algo:ABCMoC} is
$$
\widehat{B_{12}}(\by) = \dfrac{\pi(\mathcal{M}=2)}{\pi(\mathcal{M}=1}\,
\dfrac{\sum_{i=1}^N \mathbb{I}_{m^{(i)}=1)}}{\sum_{i=1}^N \mathbb{I}_{m^{(i)}=2}} 
$$
An alternative representation is given by
$$
\widehat{B_{12}}(\by) = \dfrac{\pi(\mathcal{M}=2)}{\pi(\mathcal{M}=1)}\,\dfrac{\sum_{t=1}^T 
\mathbb{I}_{m^{t}=1}\,\mathbb{I}_{\rho\{\feta(\bz^t),\feta(\by)\}\le \epsilon}}{\sum_{t=1}^T
\mathbb{I}_{m^{t}=2}\,\mathbb{I}_{\rho\{\feta(\bz^t),\feta(\by)\}\le \epsilon}}\,,
$$
where the pairs $(m^t,z^t)$ are simulated from the (joint) prior and $T$ is the total number of simulations
that are necessary for $N$ acceptances in Algorithm \ref{algo:ABCMoC}. In order to study the limiting behaviour of this
approximation, we first let $T$ go to infinity. (For simplification purposes and without loss of generality, we
choose a uniform prior on the model index.) The limit of $\widehat{B_{12}}(\by) $ is then
\begin{eqnarray*}
B_{12}^\epsilon(\by) &=& \dfrac{\mathbb{P}[\mathcal{M}=1,\rho\{\feta(\bz),\feta(\by)\} \le \epsilon]}
                             {\mathbb{P}[\mathcal{M}=2,\rho\{\feta(\bz),\feta(\by)\} \le \epsilon]}\\
                     &=& \dfrac{\int \mathbb{I}_{\rho\{\feta(\bz),\feta(\by)\} \le \epsilon} 
\pi_1(\btheta_1)f_1(\bz|\btheta_1)\,\text{d}\bz\,\text{d}\btheta_1}
                             {\int \mathbb{I}_{\rho\{\feta(\bz),\feta(\by)\} \le \epsilon} 
\pi_2(\btheta_2)f_2(\bz|\btheta_2)\,\text{d}\bz\,\text{d}\btheta_2}\\
                     &=& \dfrac{\int \mathbb{I}_{\rho\{\feta,\feta(\by)\} \le \epsilon} 
\pi_1(\btheta_1)f_1^{\feta}(\feta|\btheta_1)\,\text{d}\feta\,\text{d}\btheta_1}
                             {\int \mathbb{I}_{\rho\{\feta,\feta(\by)\} \le \epsilon} 
\pi_2(\btheta_2)f_2^{\feta}(\feta|\btheta_2)\,\text{d}\feta\,\text{d}\btheta_2}\,,
\end{eqnarray*}
where $f_1^{\feta}(\feta|\btheta_1)$ and $f_2^{\feta}(\feta|\btheta_2)$ denote the distributions of $\feta(\bz)$ when
$\bz\sim f_1(\bz|\btheta_1)$ and $\bz\sim f_2(\bz|\btheta_2)$, respectively. By L'Hospital formula, if we let $\epsilon$
go to zero, the above converges to
$$
B^{\feta}_{12}(\by)=\dfrac{\int \pi_1(\btheta_1)f_1^{\feta}(\feta(\by)|\btheta_1)\,\text{d}\btheta_1}
                          {\int \pi_2(\btheta_2)f_2^{\feta}(\feta(\by)|\btheta_2)\,\text{d}\btheta_2}\,,
$$
which is precisely and exactly the Bayes factor for testing model $1$ versus model $2$ based on the sole observation of 
$\feta(\by)$. This result is completely coherent with the current perspective on ABC, namely that the inference derived
from the ideal ABC output when $\epsilon=0$ only uses the information contained in $\feta(\by)$. Thus, in the limiting case, i.e.~when
the ABC algorithm uses an infinite computing power, the ABC odds ratio does not take into account the features of the data besides
the value of $\feta(\by)$, which is why the limiting Bayes factor only depends on the distributions of $\feta$ under both models. 

In contrast with point estimation---where using a sufficient statistic has no impact on the inference in the limiting case---, 
the loss of information resulting from considering solely $\feta$ seriously impacts the resulting inference on
which model is best supported by the data.  Indeed, as exhibited in a special case
by \cite{grelaud:marin:robert:rodolphe:tally:2009}, the information contained in $\feta(\by)$ is
almost always smaller than the information contained in $\by$ and this even in the case $\feta(\by)$ is a
sufficient statistic for {\em both models}. In other words, {\em $\feta(\by)$ being sufficient for both $f_1(\by|\btheta_1)$
and $f_2(\by|\btheta_2)$ does not usually imply that $\feta(\by)$ is sufficient for $\{m,f_m(\by|\btheta_m)\}$.}
To see why this is the case, consider the most favourable case, namely when
$\feta(\by)$ is a sufficient statistic for both models. We then have by the factorisation theorem
\citep{lehmann:casella:1998} that $f_i(\by|\btheta_i) = g_i(\by) f_i^{\feta}(\feta(\by)|\btheta_i)$, therefore that
\begin{eqnarray}\label{eq:DireStraits}
B_{12}(\by) &=& \dfrac{w_1(\by)}{w_2(\by)}\nonumber\\ 
&=& \dfrac{\int_{\Theta_1} \pi(\btheta_1) g_1(\by) f_1^{\feta}(\feta(\by)|\btheta_1)\,\text{d}
\btheta_1}{\int_{\Theta_2} \pi(\btheta_2) g_2(\by) f_2^{\feta}(\feta(\by)|\btheta_2)\,\text{d} \btheta_2}\nonumber\\
&=& \dfrac{g_1(\by)\,\int \pi_1(\btheta_1)f_1^{\feta}(\feta(\by)|\btheta_1)\,\text{d}\btheta_1}
{g_2(\by)\,\int \pi_2(\btheta_2)f_2^{\feta}(\feta(\by)|\btheta_2)\,\text{d}\btheta_2}\nonumber\\ 
&=& \dfrac{g_1(\by)}{g_2(\by)}\,B^{\feta}_{12}(\by)\,.
\end{eqnarray}
Therefore, unless $g_1(\by)=g_2(\by)$, the two Bayes factors differ by this ratio, $g_1(\by)/g_2(\by)$, which is only
equal to one in a very small number of known cases. This decomposition is a straightforward
proof that a model-wise sufficient statistic is usually not sufficient across models, i.e.~for model comparison.
An immediate corollary is that the ABC-MC approximation does not converge to the exact Bayes factor.

The discrepancy between the limiting ABC inference and the genuine Bayesian inference does not completely come
as a surprise, because ABC is indeed an approximation method. Users of ABC algorithms are therefore prepared
for some degree of imprecision in their final answer, a point stressed by \cite{wilkinson:2008} or
\cite{fearnhead:prangle:2010} when they qualify ABC as exact inference on a wrong model.  However, the
magnitude of the difference between $B_{12}(\by)$ and $B^{\feta}_{12} (\by)$ expressed by
\eqref{eq:DireStraits} is such that there is no direct connection between both answers. In a general setting,
if $\feta$ has the same dimension as one component of the $n$ components of $\by$, the ratio
$g_1(\by)/g_2(\by)$ is equivalent to a density ratio for a sample of size $\text{O}(n)$, hence it can be
arbitrarily small or arbitrarily large when $n$ grows. On the opposite, the Bayes factor $B^{\feta}_{12}(\by)$
is based on what is equivalent to a single observation, hence does not necessarily converge with $n$, as shown
by the Poisson and normal examples below. The conclusion derived from one Bayes factor may therefore completely
differ from the conclusion derived from other one and there is no possibility of a generic agreement between
both, or even of a manageable correction factor.

For this reason, we conclude that the ABC approach cannot be used for testing nor for model choice, with the
exception of Gibbs random fields as explained in the next section. In all cases when $g_1(\by)/g_2(\by)$ is
different from one and impossible to approximate, no inference on the true Bayes factor can be made based on
the ABC-MC approximation without further information on the ratio $g_1(\by)/g_2(\by)$, which is most often
unavailable.

We note that \cite{didelot:everitt:johansen:lawson:2011} also derived this relation between both Bayes factors
in their formula (18) but surprisingly concluded on advocating the use of ABC in complex models, where there
are no sufficient statistics. We disagree with this perspective for reasons that will be made clear in the
following sections.

\section{The special case of Gibbs random fields}\label{TheClash}

\cite{grelaud:marin:robert:rodolphe:tally:2009} showed that, for Gibbs random fields and in particular for Potts models,
when the goal is to compare several neighbourhood structures, the computation of the posterior probabilities of the
models/structures under competition can be operated by likelihood-free simulation techniques, in the sense that there
exists a converging approximation to the true Bayes factor. The reason for this property is that, in the above ratio,
$g_1(\by)=g_2(\by)$ in this special model.

Indeed, if we consider a Gibbs random field given by the likelihood function 
$$
f(\by|\btheta)=\dfrac{1}{Z_{\btheta}}\exp\{\btheta^\text{T}\eta(\by)\}\,, 
$$
where $\by$ is a vector of dimension $n$ taking values over the finite set $\mathcal{X}$ 
(possibly a lattice), $\eta(\cdot)$ is the potential function
defining the random field, taking values in $\mathbb{R}^p$, $\btheta\in\mathbb{R}^p$ is the associated parameter, and
$Z_{\btheta}$ is the corresponding normalising constant, the potential function $\eta$ is a sufficient statistic for the
model. For instance, in Potts models, the sufficient statistic is the number of neighbours,
$$
\eta(\by)=\sum_{i' \sim i} \mathbb{I}_{\{y_i=y_{i'}\}}\,,
$$
associated with a neighbourhood structure denoted by $i \sim i'$ (meaning that $i$ and $i'$ are neighbours).

The property that validates an ABC resolution for the comparison of Gibbs random fields is that, due to their specific structure, there
exists a sufficient statistic vector that runs across models and which allows for an exact (when $\epsilon=0$)
simulation from the posterior probabilities of the models. More specifically, consider $M$ Gibbs random fields in
competition, each one being associated with a potential function $\eta_m$ $(1\le m\le M)$, i.e.~with corresponding
likelihood 
$$ 
f_m(\by|\btheta_m)=\exp \left\{ \btheta_m^\text{T} \eta_m(\by) \right\} \big/ Z_{\btheta_m,m}\,,
$$ where $\btheta_m\in\Theta_m$ and $Z_{\btheta_m,m}$ is the unknown normalising constant.  A Bayesian analysis operates on
the extended parameter space $\Theta=\cup_{m=1}^{M}\{m\}\times\Theta_{m}$ that includes both the model index
$\mathcal{M}$ and the corresponding parameter space $\Theta_m$. The inferential target is thus the model posterior probability 
$$
\mathbb{P}(\mathcal{M}=m|\by)\propto\int_{\Theta_m} f_m(\by|\btheta_m) \pi_m(\btheta_m) \,\text{d}\btheta_m\,\pi(\mathcal{M}=m)\,, 
$$ 
i.e.~the marginal in $\mathcal{M}$ of the posterior distribution on $(\mathcal{M},\btheta_1,\ldots,\btheta_M)$ given
$\by$.  Each model has its own sufficient statistic $\eta_m(\cdot)$. Then, for {\em each} model, the
vector of statistics $\feta(\cdot)=(\eta_1(\cdot),\ldots,\eta_M(\cdot))$ is clearly sufficient; furthermore
\cite{grelaud:marin:robert:rodolphe:tally:2009} exposed the fact that $\feta$ is also sufficient for the joint parameter
$(\mathcal{M},\btheta_1,\ldots,\btheta_M)$. That this concatenation of sufficient statistics is 
jointly sufficient across models is a property that is rather specific to Gibbs random field models, at least from
a practical perspective (see below). Figure
\ref{figGrrr} shows an experiment from \cite{grelaud:marin:robert:rodolphe:tally:2009} concluding rightly at the
agreement between the exact Bayes factor and an ABC approximation.
\begin{figure}[H]
\centering
\includegraphics[width=.6\textwidth]{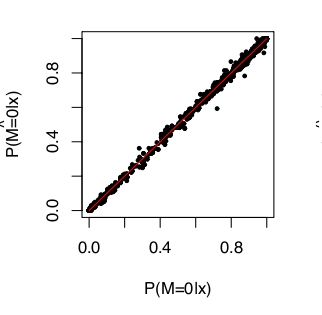}
\caption{\label{figGrrr}Comparison between the true Bayes factor and the ABC approximation
in a Markov model selection of \cite{grelaud:marin:robert:rodolphe:tally:2009}, based on
$2,000$ simulated sequences and $4\times 10^6$ proposals from the prior. The solid/red line is the diagonal.
{\em (Source: \citealp{grelaud:marin:robert:rodolphe:tally:2009}.)}}
\end{figure}

\cite{didelot:everitt:johansen:lawson:2011} point out that this specific property of Gibbs random fields can be
extended to any exponential family (hence to any setting enjoying sufficient statistics, see
e.g.~\citealp{casella:berger:1990}). Their argument is an encompassing property: by including all sufficient
statistics and all dominating measure statistics in an encompassing model, models under comparison become
submodels of the encompassing model. They then conclude that the concatenation of those statistics is jointly
sufficient across models. While this encompassing principle holds in full generality, in particular when
comparing models that are already embedded, we think it leads to a biased perspective about the merits of ABC
for model choice: in practice, complex models do not enjoy sufficient statistics (if only because they are not
exponential families). As demonstrated in the next section, there is more than a mere loss of information due
to the use of insufficient statistics and looking at what happens in the limiting case when one is relying on a
common sufficient statistic is a formal study that brings light on the potentially huge discrepancy between the
ABC-based Bayes factor and the true Bayes factor. To study a solution to the problem in the formal case of the
exponential families does not help in the understanding of the discrepancy in non-exponential models.

\section{Arbitrary ratios}\label{Siouxie}

The difficulty with the arbitrary discrepancy between $B_{12}(\by)$ and $B^{\feta}_{12}(\by)$ is that it is
impossible to evaluate in a general setting, while there is no reason to expect a reasonable agreement between
both quantities. A first illustration was produced by \cite{marin:pudlo:robert:ryder:2011} in the setting of
$MA(q)$ time series: a simulation experiment showed that, when comparing an $MA(2)$ with an $MA(1)$ model, the
ABC approximation to the Bayes factor was stable (around $2.3$) as $\epsilon$ decreases, remaining far from the
true Bayes factor $17.7$ for an $MA(2)$ simulated sample, while the approximation was $0.25$ against a true
value of $0.004$ in the case of a simulated $MA(1)$ sample. 

\subsection{A Poisson-negative binomial illustration}
As a first illustration of the discrepancy due to the use of a sufficient statistic, consider the simple case
when a sample $\by=(y_1,\ldots,y_n)$ could come from either a Poisson $\mathcal{P}(\lambda)$ distribution or
from a geometric $\mathcal{G}(p)$ distribution, already introduced in
\cite{grelaud:marin:robert:rodolphe:tally:2009} as a counter-example to Gibbs random fields and later
reprocessed in \cite{didelot:everitt:johansen:lawson:2011} to support their sufficiency argument.  In this
setting, the sum $S=\sum_{i=1}^n y_i=\feta(\by)$ is a sufficient statistic for both models but not across models. The
distribution of the sample given $S$ is a multinomial $\mathcal{M}(S,1/n,\ldots,1/n)$ distribution when the
data is Poisson, since $S$ is then a Poisson $\mathcal{P}(n\lambda)$ variable, while it is the uniform
distribution with constant probability
$$
\dfrac{1}{{n+S-1\choose S}} \mathbb{I}_{\sum_i y_i=S} = \dfrac{S!(n-1)!}{(n+S-1)!} \mathbb{I}_{\sum_i y_i=S}  
$$
in the geometric case, since $S$ is then a negative binomial $\mathcal{N}eg(n,p)$ variable. The discrepancy
ratio is therefore
$$
\dfrac{g_1(\by)}{g_2(\by)} = \dfrac{S!n^{-S}/\prod_i y_i!}{1\big/{n+S-1\choose S}}
$$
When simulating $n$ Poisson or geometric variables and using prior distributions
$$
\lambda \sim \mathcal{E}(1)\,,\quad p\sim\mathcal{U}(0,1)\,,
$$
on the respective models, the exact Bayes factor can be evaluated and the range and distribution of the
discrepancy are therefore available. Figure \ref{fig:poisneg} gives the range of $B_{12}(\by)$ versus
$B^{\feta}_{12}(\by)$, showing that $B^{\feta}_{12}(\by)$ is in this case absolutely un-related with
$B_{12}(\by)$: The values produced by both approaches simply have nothing in common. As noted above, the
approximation $B^{\feta}_{12}(\by)$ based on the sufficient statistic $S$ is producing figures of the magnitude
of a {\em single} observation, while the true Bayes factor is of the order of the sample size.
\begin{figure}[hbtp]
\centering
\includegraphics[width=.6\textwidth]{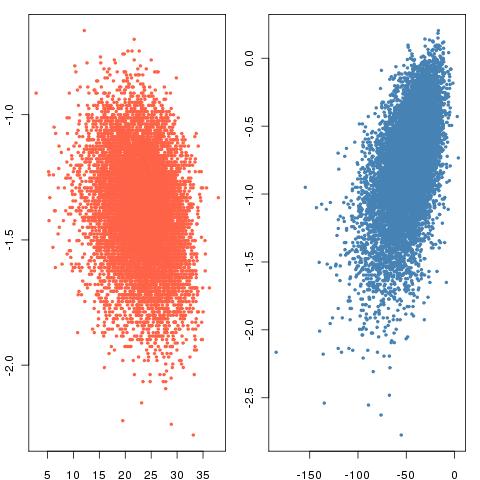}
\caption{\label{fig:poisneg}Comparison between the true log-Bayes factor {\em (first axis)} for the comparison of a Poisson
model versus a negative binomial model and of the log-Bayes factor based on the sufficient
statistic $\sum_i y_i$ {\em (second axis)}, for Poisson {\em (left)} and negative binomial {\em (left)}
samples of size $n=50$, based on $T=10^4$ replications.}
\end{figure}

The discrepancy between both Bayes factors is in fact increasing with the sample size, as shown by the
following result:

\begin{theorem}
Consider performing model selection between model 1: $\mathcal{P}(\lambda)$ with prior
distribution $\pi_1(\lambda)$ equal to an $\mathcal{E}(1)$ distribution and
model 2: $\mathcal{G}(p)$ with a uniform prior distribution $\pi_2$ when the observed data $\by$ consists of
iid observations with $\mathbb{E}[y_i] = \theta_0 > 0$. Then $S(\by) = \sum_{i=1}^n y_i$ is the
minimal sufficient statistic for both models and the Bayes factor based on the sufficient statistic $S(\by)$,
$B^{\feta}_{12}(\by)$, satisfies
$$ 
\lim_{n \rightarrow \infty} B^{\feta}_{12}(\by)= \dfrac{(\theta_0+1)^2}{\theta_0} e^{-\theta_0}\quad \mbox{a.s.}
$$
\end{theorem}

Therefore, the Bayes factor based on the sufficient statistic $S(\by)$ is \emph{not} consistent; it converges to a
non-zero, finite value almost surely.

\begin{proof}
Under model 1, we have $S \sim \mathcal{P}(n\lambda)$, with corresponding likelihood
$$f^S_1(S|\lambda) = \dfrac{1}{\Gamma(S+1)}{(n\lambda)^{S} e^{-n\lambda}} \;.$$
The marginal likelihood of $S$ under the prior $\pi_1$ is then
\begin{eqnarray}
\int_{0}^\infty  \dfrac{\lambda^{S} e^{-n\lambda}}{\Gamma(S+1) \, n^{-S}} \,e^{-\lambda} \,\text{d}\lambda &=&
\dfrac{1}{S}\int_{0}^\infty \dfrac{\lambda^{S} e^{-(n+1)\lambda} }{\Gamma(S) \, n^{-S}}  \,\text{d}\lambda  \nonumber \\
&=& \dfrac{1}{S} {n^S \over (n+1)^S} =  \dfrac{1}{S} \Big(1 + {1\over n} \Big)^{-S}  \;.\label{eqn:po}
\end{eqnarray}
Under model 2, the sufficient statistic has a negative binomial distribution, $S \sim \mathrm{Neg}(n,p)$ and thus
$$f^S_2(S|p) = {n+S-1 \choose S} p^S(1-p)^n = \dfrac{\Gamma(S+n)}{ \Gamma(S+1)\,\Gamma(n)} p^S(1-p)^n \;.$$
The corresponding marginal likelihood under the prior $\pi_2$ is
\begin{eqnarray}
 \dfrac{\Gamma(S+n)}{ \Gamma(S+1)\,\Gamma(n)} \int_{0}^1 p^{S}(1-p)^{n}  \text{d}p 
 &=&  \dfrac{\Gamma(S+n)}{ \Gamma(S+1)\,\Gamma(n)}\,\mathrm{Beta}(S+1, n+1) \nonumber \\
 &=& \dfrac{n}{(S+n+1)(S+n)} \;.\label{eqn:ge}
\end{eqnarray}
Therefore from \eqref{eqn:po} and \eqref{eqn:ge}, the Bayes factor based on the sufficient statistic is given by
\begin{equation} \label{eqn:BFS}
B^{\feta}_{12}(\by) =  \Big(1 + \dfrac{1}{n} \Big)^{-S} \times \dfrac{(S+n)\,(S+n+1)}{S\,n}
\end{equation}
Since the $y_i$'s are iid with mean $\theta_0$,  the Law of Large Numbers implies that
$S/n \rightarrow \theta_0$ almost surely, thus
$$ 
\lim_{n \rightarrow \infty} \dfrac{(S+n)\,(S+n+1)}{S\,n} =  \dfrac{(\theta_0+1)^2}{\theta_0} 
$$
since $\theta_0 > 0$. Furthermore,
$$ 
\lim_{n\rightarrow \infty} \Big(1 + {1\over n} \Big)^{-S} =
\lim_{n\rightarrow \infty} e^{-S \log(1+1/n)} = e^{-\theta_0}.
$$
Thus from \eqref{eqn:po}--\eqref{eqn:BFS} we deduce that
$$ 
\lim_{n\rightarrow \infty} {B}^{\feta}_{12}(\mathbf{y}) = e^{-\theta_0} \,  \dfrac{(\theta_0+1)^2}{\theta_0} 
$$
proving the result.
\end{proof}

In this specific setting, \cite{didelot:everitt:johansen:lawson:2011} show that adding $P=\prod_i y_i!$ to the
sufficient statistic $S$ induces a statistic $(S,P)$ that is sufficient across both models. While this is a
mathematically correct observation, we think it is not helpful for the understanding of the behaviour of
ABC-model choice in realistic settings: outside toy examples as the one above and well-structured although complex
exponential families like Gibbs random fields, it is not possible to come up with completion mechanisms that
ensure sufficiency across models and it is therefore more fruitful to consider the diverging behaviour of the
ABC approximation as given, rather than attempting at solving the problem.

\subsection{A normal illustration}

First, note that, given a one-dimensional sufficient statistic $S=\feta(\by)$, the functions
$g_1(\by)$ and $g_2(\by)$ can on principle be anything. For instance, 
$$
g_1(\by) = \prod_{i=1}^n \varphi(y_i-S|\sigma_1^2) \,\mathbb{I}_{\sum_i y_i = nS}
$$
and
$$
g_2(\by) = \prod_{i=1}^n \varphi(y_i-S|\sigma_2^2) \,\mathbb{I}_{\sum_i y_i = nS}
$$
is a possible model.
In other words, by a reparameterisation of the models, we could observe $\by=(y_1,\ldots,y_{n-1},S)$ with
$$
y_1,\ldots,y_{n-1}|S \stackrel{\text{iid}}{\sim}\mathcal{N}(S,\sigma_1^2)
\qquad\text{and}\qquad
y_1,\ldots,y_{n-1}|S \stackrel{\text{iid}}{\sim}\mathcal{N}(S,\sigma_2^2)\,,
$$
this independently of the distributions of $S$ under both models. (This means that we can find two competing
models where the distributions of $S$ are not connected with $\sigma_1$ nor with $\sigma_2$.) Because they
depend on the choice of those distributions, the true Bayes factor and the ABC-Bayes factor are unrelated and
may as well diverge from one another.  Admitedly, this construct is artificial in that there is no clear
statistical setting when this could occur, but the construct is both mathematically valid and informative about
the lack of control over the diverging factor $g_1(\by)/g_2(\by)$.

If we look at a fully normal $\mathcal{N}(\mu,\sigma^2)$ setting, we have
$$
f(\by|\mu) \propto \exp\left\{-n\sigma^{-2}(\bar y - \mu)^2/2 -\sigma^{-2} \sum_{i=1}^n (y_i-\bar y)^2 /2\right\} \sigma^{-n}
$$
hence
$$
f(\by | \bar y) \propto \exp\left\{-\sigma^{-2} \sum_{i=1}^n (y_i-\bar y)^2 /2\right\} \sigma^{-n} \mathbb{I}_{\sum y_i = n\bar y}\,.
$$
If we reparameterise the observations into $\bu = (y_1-\bar y,\ldots,y_{n-1}-\bar y,\bar y)$, we do get
\begin{eqnarray*}
f(\bu|\mu) &\propto&  \sigma^{-n}\, \exp\left\{-n\sigma^{-2}(\bar y - \mu)^2/2 \right\}\\
    &    & \times\exp\left\{-\sigma^{-2} \sum_{i=1}^{n-1} u_i^2/2 - 
             \sigma^{-2} \left[ \sum_{i=1}^{n-1} u_i \right]^2  \big/2 \right\}
\end{eqnarray*}
since the Jacobian is $1$. Hence
$$
f(\bu|\bar y) \propto \exp\left\{-\sigma^{-2} \sum_{i=1}^{n-1} u_i^2/2 - 
    \sigma^{-2} \left[ \sum_{i=1}^{n-1} u_i \right]^2  /2 \right\} \sigma^{-n} 
$$

Considering both models
$$
y_1,\ldots,y_n \stackrel{\text{iid}}{\sim} \mathcal{N}(\mu,\sigma_1^2)\quad\text{ and }\quad
y_1,\ldots,y_n \stackrel{\text{iid}}{\sim} \mathcal{N}(\mu,\sigma_2^2)\,,
$$
the discrepancy ratio is then given by
\begin{align*}
\dfrac{g_1(\by)}{g_2(\by)}&= \dfrac{\exp\left\{-\sigma_1^{-2} \sum_{i=1}^{n-1} (y_i-\bar y)^2/2 -
    \sigma_1^{-2} \left[ \sum_{i=1}^{n-1} (y_i-\bar y) \right]^2  /2 \right\} \sigma_1^{-n+1}}
{\exp\left\{-\sigma_2^{-2} \sum_{i=1}^{n-1} (y_i-\bar y)^2/2 -
    \sigma_2^{-2} \left[ \sum_{i=1}^{n-1} (y_i-\bar y) \right]^2  /2 \right\} \sigma_2^{-n+1}} \\
&\quad= \dfrac{\sigma_2^{n-1}}{\sigma_1^{n-1}}\,
\exp\left\{\dfrac{\sigma_2^{-2}-\sigma_1^{-2}}{2}\left(\sum_{i=1}^{n-1} (y_i-\bar y)^2+ \left[ \sum_{i=1}^{n-1}
(y_i-\bar y) \right]^2  \right)\right\} 
\end{align*}
and is connected with the lack of consistency of the Bayes factor:

\begin{theorem}
Consider performing model selection between model 1: $\mathcal{N}(\mu,\sigma_1^2)$ and model 2:
$\mathcal{N}(\mu,\sigma_2^2)$, $\sigma_1$ and $\sigma_2$ being given, with prior distributions
$\pi_1(\mu)=\pi_2(\mu)$ equal to a $\mathcal{N}(0,a^2)$ distribution and when the observed data $\by$ consists
of iid observations with finite mean and variance.  Then $S(\by) = \sum_{i=1}^n y_i$ is the minimal sufficient
statistic for both models and the Bayes factor based on the sufficient statistic $S(\by)$,
$B^{\feta}_{12}(\by)$, satisfies
$$
\lim_{n \rightarrow \infty} B^{\feta}_{12}(\by)= 1
\quad\mbox{a.s.}
$$
\end{theorem}

\begin{proof}
The marginal likelihood associated with $S(\by)$ and the prior $\mu\sim\mathcal{N}(0,a^2)$ is 
\begin{align*}
m^{\feta}(S) &\propto \sqrt{n}\sigma_1^{-1} \int e^{-n(\bar y-\mu)^2/2\sigma_1^2} e^{-\mu^2/2a^2} \,\text{d}\mu \\
&= \sqrt{n}\sigma_1^{-1} \exp\left\{ -\dfrac{\bar y^2}{2(a^2+\sigma_1^2/n)} \right\} \big/ \sqrt{n\sigma_1^{-1}+a^{-2}}\,,
\end{align*}
hence leading to the Bayes factor
$$
B^{\feta}_{12}(\by) = \dfrac{\sigma_2}{\sigma_1}\,\dfrac{\exp\left\{ -\dfrac{\bar y^2}{2(a^2+\sigma_1^2/n)}
\right\}}{\exp\left\{ -\dfrac{\bar y^2}{2(a^2+\sigma_2^2/n)}
\right\}}\,\dfrac{\sqrt{n\sigma_2^{-1}+a^{-2}}}{\sqrt{n\sigma_1^{-1}+a^{-2}}}\,,
$$
which indeeds converges to $1$ as $n$ goes to infinity.
\end{proof}

Figure \ref{fig:twonormal} illustrates the behaviour of the discrepancy ratio when $\sigma_1=0.1$ and $\sigma_2=10$, for datasets of size
$n=15$ simulated according to both models. The discrepancy (expressed on a log scale) is once again dramatic, in concordance with the above lemma.

\begin{figure}[hbtp]
\centering
\includegraphics[width=.6\textwidth]{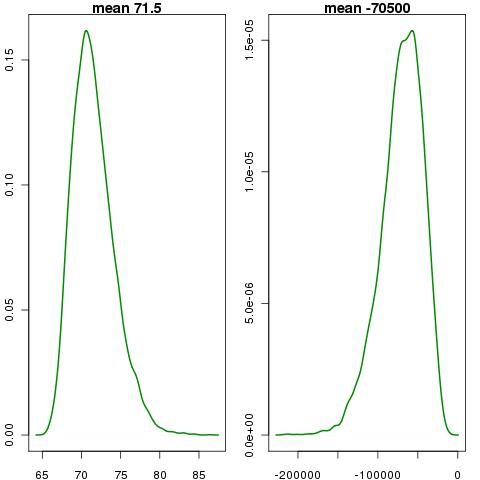}
\caption{\label{fig:twonormal}Empirical distributions of the log discrepancy $\log g_1(\by)/ g_2(\by)$ for 
datasets of size $n=15$ simulated from $\mathcal{N}(\mu,\sigma_1^2)$ {\em (left)} and $\mathcal{N}(\mu,\sigma_2^2)$ {\em (right)}
distributions when $\sigma_1=0.1$ and $\sigma_2=10$, based on $10^4$ replications and a flat prior.}
\end{figure}

If we now turn to an alternative choice of sufficient statistic, using the pair $(\bar y,S^2)$ with
$$
S^2=\sum_{i=1}^n (y_i-\bar y)^2\,,
$$
we follow the solution of \cite{didelot:everitt:johansen:lawson:2011}. Using a conjugate prior
$\mu\sim\mathcal{N}(0,a^2)$, the true Bayes factor is given by 
$$
B_{12}(\by) = \dfrac{\sigma_1^{-n}}{\sigma_2^{-n}}\,\dfrac{\exp\{-S^2/2\sigma_1^2\}}{\exp\{-S^2/2\sigma_1^2\}}\,
\dfrac{\exp\{-\bar y^2/2(a^2+\sigma_1^2/n)\}}{\exp\{-\bar y^2/2(a^2+\sigma_2^2/n)\}}\,
\dfrac{\sqrt{a^{-2}+\sigma_2^{-2}n}}{\sqrt{a^{-2}+\sigma_1^{-2}n}}\,.
$$
and it is equal to the Bayes factor based on the corresponding distributions of the pair $(\bar y,S^2)$ in the respective models. 
Again, we do not think this coincidence brings the proper light on the behaviour of the ABC approximations in
realistic settings.

\section{Conclusion}\label{RollingStones}

Since its introduction by \cite{tavare:balding:griffith:donnelly:1997} and
\cite{pritchard:seielstad:perez:feldman:1999}, ABC has been extensively used in several areas involving complex
likelihoods, primarily in population genetics.  In those domains, ABC has been used both for point estimation
and testing of hypotheses. In realistic settings, with the exception of Gibbs random fields that satisfy a
resilience property with respect to their sufficient statistics, the conclusions drawn on model comparison
cannot alas be trusted {\em per se} but require further analyses as to the pertinence of the (ABC) Bayes factor
based on the summary statistics.  This paper has only examined in details the case when the summary statistics
are sufficient for both models, while practical situations imply the use of in-sufficient statistics, and
further research is needed for the latter case. However, this practical situation implies a wider loss of
information compared with the exact inferential approach, hence a wider discrepancy between the exact Bayes
factor and the quantity produced by an ABC approximation. It thus appears to us an urgent duty to warn the
community about the dangers of this approximation, especially when considering the rapidly increasing number of
applications using ABC for conducting model choice and hypothesis testing. As a final (and negative) point, we
unfortunately do not see an immediate and generic alternative for the approximation of Bayes factors because
importance sampling techniques are suffering from the same difficulty, namely they only depend on the summary
statistics.  

As a final remark, we note that \cite{sousa:etal:2009} advocate the use of full allelic distributions in an ABC
framework, instead of resorting to summary statistics. They show that it is possible to apply ABC using allele
frequencies to draw inferences in cases where it is difficult to select a set of suitable summary statistics
(and when the complexity of the model or the size of dataset makes it computationally prohibitive to use
full-likelihood methods). In such settings, were we to consider a model choice problem,
the divergence exhibited in the current paper would not occur
because the measure of distance does not rely on a reduction of the sample.


\section*{Acknowledgements}

The first two authors' work has been partly supported by the Agence Nationale
de la Recherche (ANR, 212, rue de Bercy 75012 Paris) through the 2009-2012
project {\sf Emile}, directed by Jean-Marie Cornuet.  

\small

\end{document}